\documentclass[10pt, oneside]{article}   	% use "amsart" instead of "article" for AMSLaTeX format
\usepackage[english]{babel}
\usepackage[toc,page]{appendix}
\usepackage[margin=1in]{geometry}  
\usepackage{hyperref}% See geometry.pdf to learn the layout options. There are lots.
\usepackage{graphicx,wrapfig,lipsum}	
\usepackage[labelformat=empty]{caption}
\usepackage{yfonts}	

\usepackage[lined,boxed,ruled,norelsize,algo2e,linesnumbered]{algorithm2e}
\usepackage{thmtools}
\usepackage{xspace}

\usepackage{amssymb}
\usepackage{amsmath}
\usepackage{amsthm}
\usepackage{todonotes}

\newcommand{\cost}{\text{cost}}
\newcommand{\dd}{\textrm {d}}
\newcommand{\Supp}{\text{Supp}}

\def\Xint#1{\mathchoice
{\XXint\displaystyle\textstyle{#1}}%
{\XXint\textstyle\scriptstyle{#1}}%
{\XXint\scriptstyle\scriptscriptstyle{#1}}%
{\XXint\scriptscriptstyle\scriptscriptstyle{#1}}%
\!\int}
\def\XXint#1#2#3{{\setbox0=\hbox{$#1{#2#3}{\int}$ }
\vcenter{\hbox{$#2#3$ }}\kern-.6\wd0}}

\def\dashint{\Xint-}

\newtheorem{thm}{Theorem}[section]
\newtheorem{lem}[thm]{Lemma}
\newtheorem{prop}[thm]{Proposition}
\newtheorem{cor}[thm]{Corollary}

\theoremstyle{definition}
\newtheorem{defn}{Definition}
\newtheorem*{rem}{Remark}

\title{Chasing Convex Bodies Optimally}
\author{Mark Sellke\\
Stanford University}

\date{}

\begin{document}

\maketitle

\abstract{In the \emph{chasing convex bodies} problem, an online player receives a request sequence of $N$ convex sets $K_1,\dots, K_N$ contained in a normed space $X$ of dimension $d$. The player starts at $x_0=0\in X$, and at time $n$ observes the set $K_n$ and then moves to a new point $x_n\in K_n$, paying a cost $||x_n-x_{n-1}||$. The player aims to ensure the total cost exceeds the minimum possible total cost by at most a bounded factor $\alpha_d$ independent of $N$, despite $x_n$ being chosen without knowledge of the future sets $K_{n+1},\dots,K_N$. The best possible $\alpha_d$ is called the competitive ratio. Finiteness of the competitive ratio for convex body chasing was proved for $d=2$ in \cite{friedmanlinial} and conjectured for all $d$. \cite{chasingconvex2018} recently resolved this conjecture, proving an exponential $2^{O(d)}$ upper bound on the competitive ratio. 

We give an improved algorithm achieving competitive ratio $d$ in any normed space, which is \emph{exactly} tight for $\ell^{\infty}$. In Euclidean space, our algorithm also achieves competitive ratio $O(\sqrt{d\log N})$, nearly matching a $\sqrt{d}$ lower bound when $N$ is subexponential in $d$. Our approach extends that of \cite{chasingnested2018} for \emph{nested} convex bodies, which is based on the classical Steiner point of a convex body. We define the \emph{functional} Steiner point of a convex function and apply it to the associated work function.}

\section{Introduction} % need not actually have a section label

Let $X$ be a $d$-dimensional normed space and $K_1,K_2,\dots, K_N\subseteq X$ a finite sequence of convex bodies. In the \emph{chasing convex bodies} problem, a player starting at $x_0=0\in X$ learns the sets $K_n$ one at a time, and after observing $K_n$ moves to a point $x_n\in K_n$. The player's cost is the total path length 
\[
	\cost(x_1,\dots,x_N)=\sum_{n=1}^N ||x_n-x_{n-1}||.
\] 
Denote the smallest cost (in hind-sight) among all such sequences by 
\[
	\cost(K_1,\dots,K_N)=\min_{(y_n\in K_n)_{n\leq N}}\sum_{n=1}^N ||y_n-y_{n-1}||.
\] 
The player's goal is to ensure that 
\begin{equation}\label{eq:CR}
	\cost(x_1,\dots,x_N)\leq \alpha_d\cdot \cost(K_1,\dots,K_N)
\end{equation} 
holds for any sequence $K_1,\dots,K_N$, where $\alpha_d$ is as small as possible and is independent of $N$. The challenge is that the points $x_n=x_n(K_1,\dots,K_n)$ must depend only on the sets revealed so far. To encapsulate this requirement we say the player's path must be \emph{online}, as opposed to the optimal \emph{offline} path which can depend on future information.  An online algorithm achieving \eqref{eq:CR} for some finite $\alpha_d$ is said to be $\alpha_d$-\emph{competitive}, and the smallest possible $\alpha_d$ among all online algorithms is the \emph{competitive ratio} of chasing convex bodies.

In the most general sense, the problem of asking a player to choose an online path $x_1,\dots, x_N$ through a sequence of subsets $S_1,\dots,S_N$ in a metric space $\mathcal X$ is known as \emph{metrical service systems}. These sets are typically called ``requests". When arbitrary subsets $S_i\subseteq\mathcal X$ can be requested, the competitive ratio possible is $|\mathcal X|-1$ in any metric space \cite{manasse1990competitive}. One also considers the slightly more general \emph{metrical task systems} problem in which requests are non-negative cost functions $f_n:\mathcal X\to\mathbb R^+$ rather than sets and the cost takes the form
\[
	\cost(x_1,\dots,x_N)=\sum_{n=1}^N d_{\mathcal X}(x_n,x_{n-1})+f_n(x_n)
\] 
where $\sum_{n=1}^N d_{\mathcal X}(x_n,x_{n-1})$ is called the \emph{movement cost} while $\sum_{n=1}^N f_n(x_n)$ is the \emph{service cost}. As in~\eqref{eq:CR}, one aims to ensure 
\begin{equation}\label{eq:CR2}
	\cost(x_1,\dots,x_N)\leq \alpha\cdot \cost(f_1,\dots,f_n)=\alpha\cdot \inf_{(y_n\in \mathcal X)_{n\leq N}}\cost(y_1,\dots,y_N).
\end{equation}
The competitive ratio of metrical task systems is always $2|\mathcal X|-1$ \cite{borodin1992optimal}. Actually both competitive ratios just stated are for deterministic algorithms; one may also allow external randomness, so that one chooses $x_n=x_n(S_1,\dots,S_n,\omega)$ for some random variable $\omega$ independent of the sets $S_i$. One then aims for the same guarantee as in~\eqref{eq:CR},~\eqref{eq:CR2} with the expected cost of the player on the left-hand side, for any fixed sequence $(S_1,\dots,S_N)$. With randomization the competitive ratio of metrical task or service systems sharply drops and is known to be in the range $\left[\frac{c_1\log |\mathcal X|}{\log\log|\mathcal X|},c_2\left(\log|\mathcal X|\right)^2 \right]$, and to be $\Theta(\log |\mathcal X|)$ in some specific cases \cite{bartal2005metric,bartal2006ramsey,fiat2003better,bubeck2019metrical}. However this is not the end of the story as a wide range of problems, including chasing convex bodies, result from restricting which subsets are allowed as requests. The literature on such problems is vast and includes scheduling \cite{graham1966bounds}, self-organizing lists \cite{sleator1985amortized}, efficient covering \cite{alon2003online}, safely using machine-learned advice \cite{blum2000line,kumar2018improving,lykouris2018competitive,wei2020optimal}, and the famous $k$-server problem \cite{manasse1990competitive,grove1991harmonic,koutsoupias1995k,bansal2015polylogarithmic}.

Chasing convex bodies was proposed in \cite{friedmanlinial} to study the interaction between convexity and metrical task systems. Of course the general upper bounds above are of no use as $|X|=\infty$, while the lower bounds also do not apply due to the convexity constraint. \cite{friedmanlinial} gave an algorithm with finite competitive ratio for the already non-trivial $d=2$ case and conjectured that the competitive ratio is finite for any $d\in\mathbb N$. The best known asymptotic lower bounds come from requesting the faces of a hypercube by taking $K_n=(\varepsilon_1,\varepsilon_2,\dots,\varepsilon_n)\times [-1,1]^{d-n}$ for $\varepsilon_i\in \{-1,1\}$ uniformly random and $n\leq d$. This construction implies that the competitive ratio is at least $\sqrt{d}$ in Euclidean space and at least $d$ for $X=\ell^{\infty}$ - see \cite[Lemma 5.4]{chasingnested2018} for more on lower bounds. Unlike in many competitive analysis problems, randomization is useless for chasing convex bodies and we may freely restrict attention to deterministic algorithms. This is because $\cost(x_1,\dots,x_N)$ is convex on $X^N$, and so randomized paths are no better than their (deterministic) pointwise expectations. 

Following a lack of progress on the full conjecture, restricted cases such as chasing subspaces were studied, e.g. \cite{2016chasing}. A notable restriction is chasing \emph{nested} convex bodies, where the convex sets $K_1\supseteq K_2\supseteq \dots $ are required to decrease. Nested chasing was introduced in \cite{nestedarechaseable} and solved rather comprehensively in \cite{ABCGL} and then \cite{chasingnested2018}. The latter work gave an algorithm with optimal competitive ratio up to $O(\log d)$ factors for all $\ell^p$ spaces based on Gaussian-weighted centroids. Moreover it gave a $d$-competitive memoryless algorithm based on the Steiner point which we discuss later.

Some time after chasing convex bodies was posed, an equivalent problem called \emph{chasing convex functions} emerged. This is a metrical task systems problem in which requests are convex functions $f_n:X\to\mathbb R^+$ instead of convex sets. As described above the total cost
\[
	\cost(x_1,\dots,x_N)=\sum_{n=1}^N ||x_n-x_{n-1}||+f_n(x_n)
\] 
decomposes as a movement cost plus a service cost. Chasing convex functions subsumes chasing convex bodies by replacing the body $K_n$ with the function $f_n=2\cdot d(x,K_n)$. This is because an arbitrary algorithm for the requests $f_n$ is improved by projecting $x_n$ onto $K_n$ - actually the same argument shows more generally that metrical task systems subsumes metrical service systems. Conversely as observed in \cite{chasingconvex2018}, convex function chasing in $X$ can be reduced to convex body chasing in $X\oplus \mathbb R$ up to a constant factor by alternating requests of the epigraphs $\{(x,y)\in X\times \mathbb R: y\geq f_n(x)\}$ with the hyperplane $X\times \{0\}$. As with chasing convex bodies, randomized algorithms are no better than deterministic algorithms since $\cost(x_1,\dots,x_N)$ remains convex on $X^N$.

Convex function requests allow one to model many practical problems. Indeed chasing convex functions was originally considered as a model for efficient power management in cooling data centers \cite{datacenters}. In light of this, restricted or modified versions of chasing convex function have also been studied. For example, \cite{1dSOCO} determines the exact competitive ratio in $1$ dimension, while works such as \cite{SOCObalanced,goel2019beyond} show dimension-independent competitive ratios for similar problems with further restrictions on the cost functions.

\paragraph{Main Result}

In prior joint work with S. Bubeck, Y.T. Lee, and Y. Li \cite{chasingconvex2018} we gave the first algorithm achieving a finite competitive ratio for convex body chasing. Unfortunately this algorithm used an induction on dimension that led to a exponential competitive ratio $2^{O(d)}$. We give an upper bound of $d$ for the competitive ratio of chasing convex bodies in a general normed space, which is tight for $\ell^{\infty}$. In Euclidean space, our algorithm has competitive ratio $O(\sqrt{d\log N})$, nearly matching the lower bound $\sqrt{d}$ when the number of requests $N$ is sub-exponential in $d$. The statement following combines Theorems~\ref{thm:linearCR} and~\ref{thm:dlogN}.

\begin{thm}
\label{thm:main}
In any $d$-dimensional normed space there is a $d+1$ competitive algorithm for chasing convex functions and a $d$ competitive algorithm for chasing convex bodies. Moreover in Euclidean space this algorithm is $O(\sqrt{d\log N})$-competitive.

\end{thm}

The proof is inspired by our joint work with S. Bubeck, B. Klartag, Y.T. Lee, and Y. Li \cite{chasingnested2018} on chasing nested convex bodies. It is shown there that moving to the new body's \emph{Steiner point}, a stable center point of any convex body defined in \cite{steineroriginal}, gives total movement at most $d$ starting from the unit ball in $d$ dimensions. (The argument in \cite{chasingnested2018} is restricted to Euclidean space but the proof works in general as we will explain.) We extend their argument by defining the \emph{functional Steiner point} of a convex function. Our algorithm follows the functional Steiner point of the so-called \emph{work function} which encodes at any time the effective total cost of all requests so far.

We remark that given the form of \eqref{eq:CR}, chasing convex bodies may be viewed as an online version of a Lipschitz selection problem. In the broadest generality, for some family $\mathcal S\subseteq 2^X$ of subsets of a set $X$, a selector takes sets $S\in\mathcal S$ to elements $s\in S$. Of course the relevant comparison for us is when $\mathcal S$ consists of all convex bodies in $X$. Continuity and Lipschitz properties of general selectors have received significant attention \cite{shvartsman1984lipshitz,shvartsman2002lipschitz,kupka2005continuous,fefferman2018sharp}. Taking the Hausdorff metric on convex sets, the Steiner point is $d$-Lipschitz in any normed space. Moreover as explained in \cite[Section 4]{steinerlipschitz}, it achieves the exact optimal Lipschitz constant (of order $\Theta(\sqrt{d})$) when $X$ is Euclidean due to a beautiful symmetrization argument. We find it appealing that this in some sense optimal Hausdorff-Lipschitz selector also solves an online version of Lipschitz selection.

Concurrently with this work, C.J. Argue, A. Gupta, G. Guruganesh, and Z. Tang obtained similar results for chasing convex bodies in Euclidean space \cite{AGGT}. Their algorithm is based on Steiner points of level sets of the work function; these turn out to be almost the same as the functional Steiner point as we show in Section~\ref{sec:unify}.

\section{Problem Setup}

\subsection{Notations and Conventions}

The variables $T,t,s$ denote real times while $N,n$ denote integer times. $\dashint_{x\in S} f(x)\dd x$ denotes the average value $\frac{\int_{x\in S} f(x)\dd x}{\int_{x\in S} 1\dd x}$ of $f(x)$ on the set $S$. Denote by $B_1\subseteq X$ the unit ball and $B_1^*\subseteq X^*$ the dual unit ball. The symbol $\partial$ denotes boundary, and $\langle \cdot,\cdot\rangle$ denotes the natural pairing between $X,X^*$.

\subsection{Continuous Time Formulation}
\label{subsec:reductions}

Our proof is more natural in continuous time, so we first solve the problem in this setting and then specialize to discrete time. In continuous time chasing convex functions, we receive a locally bounded family of non-negative convex functions $(f_t:X\to\mathbb R^+)_{t\in [0,T]}$. We assume that $f_t(x)$ is piece-wise continuous in $t$ with a locally finite set of continuities. The player constructs a bounded variation path $(x_t)$ online, so that $x_s$ depends only on $(f_t)_{t\leq s}$. We will assume $f_t$ and $x_t$ are cadlag (right-continuous with left-limits) in the time variable $t$. The cost is again the sum of movement and service costs given by
\[
\cost((x_t)_{t\in [0,T]})=\int_{0}^T f_t(x_t)+||x'_t||\dd t.
\]
Here and throughout, the integral of $||x'_t||$ is understood to mean the total variation of the path $x_t$. As before the goal is to achieve a small competitive ratio against the optimal offline path. Given a sequence $f_1,f_2,\dots, f_N$ of convex requests, one readily obtains a corresponding continuous-time problem instance by choosing, for each $t\in [0,N]$, the function $f_t=f_{n}$ for $t\in (n-1,n]$. The next proposition shows that solving this continuous problem suffices to solve the discrete problem.

\begin{prop}
\label{prop:continuoustodiscrete}

Any discrete-time instance of chasing convex function has the same offline optimal cost as its continuous-time counterpart. Meanwhile for any continuous-time online algorithm there exists a discrete-time online algorithm achieving both smaller movement and smaller service cost on every sequence of functions $f_1,\dots,f_N$.

\end{prop}

\begin{proof}

It is easy to see that the continuous and discrete time problems have the same offline optimum value. Given a solution for continous-time convex function chasing, suppose the player sees a discrete time request $f_n$. The player then computes the continuous time path $(x_t)_{t\in (n-1,n]}$ and moves to some $x_{t_n}$ with $t_n\in (n-1,n]$ and 
\[
	f_n(x_{t_n})\leq \int_{n-1}^n f_n(x_t)\dd t.
\] 
The discretized sequence $(x_{t_1},\dots, x_{t_N})$ has a smaller movement cost than the continuous path $(x_t)_{t\in [0,T]}$ because the triangle inequality implies
\begin{align*}
	\sum_{n=1}^N ||x_{t_n}-x_{t_{n-1}}||&\leq \sum_{n=1}^N \int_{t_{n-1}}^{t_n} ||x'_s||\dd s\\
	&=\int_0^{t_N}||x'_s||\dd s\\
	&\leq \int_0^N ||x'_s||\dd s.
\end{align*}
The discretized path also has smaller service cost by construction, hence the result.

\end{proof}

\section{Functional Steiner Point and Work Function}
\label{sec:setup}

We begin by recalling the definition of the Steiner point in a $d$-dimensional normed space $X$. For a convex body $K\subseteq X$ and $v\in X^*$, define 
\begin{align*}
	f_K(v)&=\arg\max_{x\in K}\langle v,x\rangle,\\
	h_K(v)&=\max_{x\in K}\langle v,x\rangle=\langle f_K(v),x\rangle.
\end{align*}
$h_K$ is commonly known as the \emph{support function} of $K$. Let $\mu$ denote the cone measure on $\partial B_1^*$, which can be sampled from by choosing a uniformly random $z\in B^*_1$ and normalizing to $\theta=\frac{z}{||z||}.$ For $\theta\in\partial B_1^*$ define $n(\theta)\in X$ to be the outward unit normal defined (for $\mu$-almost all $\theta$) by $||n(\theta)||=1$ and $\langle n(\theta),\theta\rangle = 1$. 

\begin{defn}\cite[Chapter~6]{steinerlipschitz}\label{defn:steiner}
The Steiner point $s(K)\in X$ is
\begin{align} 
	s(K)&=\dashint_{v\in B^*_1} f_K(v) \dd v.\label{eq:primal0}\\
	&=d~\dashint_{\theta\in \partial B^*_1} h_K(\theta)n(\theta) \dd \mu(\theta).\label{eq:dual0} 
\end{align}
\end{defn}

The equivalence of the two definitions follows from the divergence theorem and the identity $\nabla h_K=f_K$. The factor $d$ comes from the discrepancy in total measure of the ball and the sphere. See \cite[Chapter 6]{steinerlipschitz} for a careful derivation. 

Using Definition~\ref{defn:steiner}, the upper bound $d$ for nested chasing in \cite{chasingnested2018} immediately extends to any normed space. We recall the main result here. It is not phrased as a competitive ratio because some apriori reductions are possible in nested chasing --- roughly speaking we stay inside the unit ball $B_1$ and treat the offline optimum cost as being $1$. Note that both \eqref{eq:primal0} and \eqref{eq:dual0} are essential in the argument below.

\begin{thm}\cite[Theorem 2.1]{chasingnested2018}\label{thm:nested}
Let $B_1\supseteq K_1\supseteq K_2\supseteq \dots\supseteq K_N$ be convex bodies in $X$, with $x_n=s(K_n)$ for each $n$. Then $x_n\in K_n$ for each $n$ and
\[
	\sum_{n=1}^{N} ||x_n-x_{n-1}||\leq d.
\]
\end{thm}

\begin{proof}

It follows from \eqref{eq:primal0} that $s(K)\in K$, so it remains to estimate the total movement. For convenience take $K_0=B_1$ the unit ball so that $x_0=(0,0,\dots, 0)=s(K_0)$. From $K_n\subseteq K_{n-1}$ it follows that $h_{K_n}(\theta)\leq h_{K_{n-1}}(\theta)$ for each $n\leq N$ and $\theta\in\partial B_1^*$. Combining with~\eqref{eq:dual0} yields:
\begin{align*} 
	\sum_{n=1}^N ||s(K_n)-s(K_{n-1})||&\leq d~\dashint_{\theta\in \partial B_1^*} \sum_{n=1}^N |h_{K_n}(\theta)-h_{K_{n-1}}(\theta)| \dd\mu(\theta)\\
	&= d~\dashint_{\theta\in \partial B_1^*} \sum_{n=1}^N h_{K_{n-1}}(\theta)-h_{K_{n}}(\theta) \dd\mu(\theta)\\
	&= d~\dashint_{\theta\in \partial B_1^*} 1-h_{K_{N}}(\theta) \dd\mu(\theta)\\
	&\leq d.
\end{align*}
Here the last inequality follows from $h_{K_N}(\theta)+h_{K_N}(-\theta)\geq 0$.

\end{proof}

We now extend the definition of Steiner point to convex functions. The idea is to replace the support function by the concave conjugate (also known as the Fenchel-Legendre transform). Recall that for a convex function $W:X\to\mathbb R^+$, the concave conjugate $W^*:X^*\to \mathbb R\cup \{-\infty\}$ is defined by
\begin{equation}\label{eq:W*}
	W^*(v)=\inf_{w\in X} \left(W(w)-\langle v,w\rangle\right)
\end{equation}
Let us assume $W$ is not only convex but also $1$-Lipschitz, and that $W(w)-||w||$ is uniformly bounded. We will refer to such a $W$ as an (abstract) work function. Note $W^*(v)$ is finite whenever $||v||<1$ by the last assumption, and moreover the infimum in \eqref{eq:W*} is attained. We denote this point attaining this infimum by
\[
    v^*=\arg\min_{w\in X} \left(W(w)-\langle v,w\rangle\right),
\]
the conjugate point to $v$ with respect to $W$. It satisfies $\nabla W(v^*)=v$ and is well-defined for almost every $v\in B_1^*$ by Alexandrov's theorem. Moreover we have $\nabla W^*(v)=-v^*$. Combining this latter relation with the divergence theorem yields another identity, from which the functional Steiner point is defined. 

\begin{defn}

Let $X$ be an arbitrary $d$-dimensional normed space, and $W:X\to\mathbb R^+$ a work function as defined above. The functional Steiner point $s(W)\in X$ is:
\begin{align}
	s(W)&=\dashint_{v\in B^*_1} v^* \dd v. \label{eq:primal}\\
	&= -d~\dashint_{\theta\in \partial B^*_1} W^*(\theta)n(\theta) \dd \mu(\theta). \label{eq:dual}
\end{align}

\end{defn}

We remark that if a convex body $K$ is identified with the function $f(x)=d(x,K)$, then the definitions above agree. We call \eqref{eq:primal0},~\eqref{eq:primal} the \emph{primal} definitions and \eqref{eq:dual0},~\eqref{eq:dual} the \emph{dual} definitions.

\subsection{The Work Function}

The work function is a central object in online algorithms; in general it records the smallest cost required to satisfy an initial sequence of requests while ending in a given state. Work function based algorithms are essentially optimal among deterministic algorithms for general metrical task systems \cite{borodin1992optimal} as well as the $k$-server problem \cite{koutsoupias1995k}.

\begin{defn}\label{def:work}

Given requests $(f_s)_{s\leq t}$, the work function $W_t(x)$ is the offline-optimal cost among paths satisfying $x_t=x$:

\begin{align}
W_t(x)=&  \inf_{\substack{x_s:[0,t]\to X\\x_t=x}}||x_0||+ \int_0^t f_s(x_s)+||x'_s||\dd s\\
=& \inf_{\substack{x_s:[0,t]\to X\\x_t=x}} \cost_t(x_s).\label{eq:workcont}
\end{align}

Here we allow $x_s:[0,t]\to X$ to be any path of bounded variation, and as before interpret $\int_0^t ||x'_s||\dd s$ to mean the total variation of the path. Likewise for a discrete-time request sequence $(f_1,\dots,f_n)$, the work function $W_n(x)$ is defined as above with $f_t=f_n$ for $t\in (n-1,n]$ or more simply by
\[
	W_n(x)=\min_{x_1,\dots,x_n\in X} ||x-x_n||+\sum_{n=1}^N ||x_n-x_{n-1}||+f_n(x_n).
\] 
For a sequence $(K_1,\dots,K_n)$ of convex set requests the work function $W_n$ is defined analogously.

\end{defn}

In the case that $f_s(x)$ is piecewise constant in $s$ (which is all we need for the original discrete-time problem), the best offline continuous time strategy clearly coincides with the best offline discrete time strategy. The infimum is attained in~\eqref{eq:workcont} in general because the paths $(x_s)_{s\leq t}$ of variation at most $C$ satisfying $x_t=T$ are compact in the usual topology on cadlag functions for any $C$, and $\cost_t$ is lower semicontinuous. 

Denote by $W^*_t(\cdot)$ the concave conjugate of $W_t$, and $v^*_t$ the point with $\nabla W_t(v^*_t)=v$. We record the following proposition summarizing the properties of the work function and its dual.

\begin{prop}
\label{prop:properties}
In either discrete or continuous time, $W_t$ and $W^*_t$ satisfy:
\begin{enumerate}
    \item \label{it:work1} $W_0(x)=||x||.$ 
    \item \label{it:work2} $W^*_0(v)=0$ whenever $||v||\leq 1$.
    \item \label{it:work3} $W_t(x)$ is increasing in $t$ and is convex for all fixed $t$.
    \item \label{it:work4} $W^*_t(x)$ is increasing in $t$ and concave in $x$.
    \item \label{it:work4} $W_t(x)$ is an abstract work function.
    \item \label{it:work5} $W^*_t(v)$ is non-negative and finite whenever $||v||\leq 1.$
    \item \label{it:work6} $\cost((f_s)_{s\in[0,t]})=\min_{x\in X} W_t(x)$.
\end{enumerate}

\end{prop}

\begin{proof}

It is clear that $W_0(x)=||x||$, and that $W_t(x)$ is increasing in $t$. The computation of $W^*_0$ is clear. Convexity of $W_t(\cdot)$ holds by convexity of $\cost_t(\cdot)$ --- given paths $x_s^0:[0,t]\to X$ and $x_s^1:[0,t]\to X$ the path $x_s^q:[0,t]\to X$ given by
\[
	x_s^q=qx_s^1+(1-q)x_s^0
\]
satisfies for any $q\in [0,1]$,
\[
	\cost_t(x_s^q)\leq q\cdot\cost_t(x_s^1)+(1-q)\cdot\cost_t(x_s^0).
\]
Convexity of $W_t$ implies that $W^*_t$ is concave by general properties of the Fenchel-Legendre transform.
Because $W_t$ is increasing in $t$, the definition~\eqref{eq:W*} implies that $W^*_t$ is increasing in $t$ as well. 
It is easy to see that $W_t$ is $1$-Lipschitz; to show
\[
	W_t(x)\leq W_t(y)+||x-y||
\]
it suffices to take the lowest cost path to $y$ and then move from $y$ to $x$. Similarly $W_t(x)-||x||$ is bounded, making $W_t$ an abstract work function. It follows from this that $W^*_t(v)$ is finite when $||v||\leq 1$.

\end{proof}

\begin{lem}
\label{lem:boundondual}
For all $t$,
\begin{align*}
	\max_{||\theta||\leq 1} W^*_t(\theta)&\leq  2\cdot \min_{x}W_t(x),\\
	\dashint_{\theta\in \partial B_1^*} W^*_t(\theta) \dd\mu(\theta) &\leq \min_{x} W_t(x),\\
	\dashint_{v\in B_1} W^*_t(v) \dd v &\leq \min_{x} W_t(x).
\end{align*}
\end{lem}

\begin{proof}

Set
\[
	OPT_t=\arg\min_x W_t(x).
\]
The definition \eqref{eq:W*} of $W^*_t$ implies
\[
	W^*_t(\theta) \leq W_T(OPT_t)-\theta\cdot OPT_t.
\]
Finally
\begin{align*}
	|W_t(OPT_t)|&= \inf_{\substack{x_s:[0,t]\to X\\x_t=OPT_t}}||x_0||+ \int_0^t f_s(x_s)+||x'_s||\dd s\\
	&\geq \inf_{\substack{x_s:[0,t]\to X\\x_t=OPT_t}}||x_0||+\int_0^t ||x'_s||\dd s\\
	&\geq |OPT_t|
\end{align*}
holds where the triangle inequality was used in the last line. All assertions now follow. 

\end{proof}

We next compute the time derivative of $W^*_t(v)$ for fixed $v$ with $|v|<1$. The proof, a simple exercise, is left to the appendix.

\begin{lem}
\label{lem:derivativeWt}

For any $\varepsilon>0$ suppose $f_s(x)$ is jointly continuous in $(s,x)$ and convex in $x$ for $(s,x)\in [t,t+\varepsilon)\times X$. Then for almost all $v$ with $||v||<1$,

\[\frac{\dd}{\dd t} W^*_t(v) = f_t(v^*_t).\]

\end{lem}

\section{Linear Competitive Ratio}
\label{sec:linearCR}

Our algorithm for continuous-time convex function chasing is defined by setting $x_t=s(W_t)$. In its analysis, the primal definition~\eqref{eq:primal} controls the service cost while the dual definition~\eqref{eq:dual} controls the movement cost. 

\begin{thm}
\label{thm:linearCR}

$x_t=s(W_t)$ is $d+1$ competitive for continuous-time convex function chasing in any $d$-dimensional normed space $X$. In particular:

\begin{enumerate}
    \item The movement cost of $x_t$ is $d$-competitive:

    \[\int_{0}^T ||x'_t||\dd t \leq d\cdot\min_x W_t(x).\]

    \item The service cost of $x_t$ is $1$-competitive:

    \[\int_{0}^T f_t(x_t)\dd t \leq \min_x W_t(x).\]

\end{enumerate}

\end{thm}

Proposition~\ref{prop:continuoustodiscrete} yields an induced algorithm for chasing bodies/functions in discrete time which we call the discrete-time functional Steiner point.

\begin{cor}
\label{cor:discretelinearCR}

The discrete-time functional Steiner point is $d+1$ competitive for chasing convex functions and $d$ competitive for chasing convex bodies.

\end{cor}

\begin{proof}[Proof of Corollary~\ref{cor:discretelinearCR}]

This follows from Proposition~\ref{prop:continuoustodiscrete} and the fact that chasing convex bodies has $0$ service cost.

\end{proof}

\begin{proof}[Proof of Theorem~\ref{thm:linearCR}]

We begin with part $1$. From the dual definition~\eqref{eq:dual} of $s(W_t)$ and the fact that $W^*_t$ increases with $t$ from $W^*_0=0$,
\begin{align*}
\int_{0}^T ||x'_t||\dd t &= d\cdot \int_0^T \left\vert\left\vert \frac{\dd}{\dd t} \dashint_{\theta\in\partial B_1^*}  W_t^*(\theta)\theta  \dd\mu(\theta) \right\vert\right\vert\\
&\leq d\cdot \int_0^T   \dashint_{\theta\in\partial B_1^*}  \left\vert \frac{\dd}{\dd t} W_t^*(\theta)   \right\vert \dd\mu(\theta)\\
&=  d\cdot\dashint_{\theta\in\partial B_1^*} W^*_T(\theta) \dd\mu(\theta).
\end{align*}
Lemma~\ref{lem:boundondual} implies \[d\cdot \dashint_{\theta\in\partial B_1^*} W^*_T(\theta) \dd\mu(\theta) \leq d\min_x W_T(x).\] This completes the proof of part $1$ and we turn to part $2$. From the primal definition~\eqref{eq:primal} and convexity of $f_t$ it follows that
\[
f_t(s(W_t)) \leq \dashint_{v\in B_1^*} f_t(v^*_t)\dd v .
\]
Integrating in time and using Lemmas~\ref{lem:derivativeWt} and~\ref{lem:boundondual} yields:
\begin{align*}
\int_{0}^T f_t(s(W_t))\dd t \leq & \dashint_{v\in B_1} \int_{0}^T f_t(v_t^*) \dd t\dd\mu(\theta) \\
= & \dashint_{v\in B_1^*}\int_{0}^T \frac{\dd}{\dd t} W^*_t(v) \dd t \dd v\\
= & \dashint_{v\in B_1^*} W^*_T(v)-W^*_0(v) \dd v\\
= & \dashint_{v\in B_1^*} W^*_T(v) \dd v\\
\leq & \min_x W_T(x).
\end{align*}

\end{proof}

\begin{rem}

In the continuous time setting, only $f_t(x_t)$ and $\nabla f_t(x_t)$ are actually necessary to solve convex function chasing. This is because the player can always lower bound $f_t$ by 

\[f_t(x)\geq \tilde f_t(x)\equiv \max\left(f_t(x_t)+\langle \nabla f_t(x_t), x-x_t\rangle,0\right).\]

As $\tilde f_t(x_t)=f_t(x_t)$, by simply pretending the requests are $\tilde f_t$, any competitive algorithm can be transformed into one which only uses the values $f_t(x_t)$ and $\nabla f_t(x_t)$ and which obeys the same guarantees.

In the discrete time setting, if we are given $f_n(x_{n-1})$ and $\nabla f_n(x_{n-1})$ before choosing $x_n$, there is another source of error because $f_n(x_n)$ is totally unknown. However this error is easily controlled when the $f_n$ are uniformly Lipschitz. Let $(x_n)_{n\leq N}$ be the discrete-time functional Steiner point sequence for the functions recursively defined by
\[
	\tilde f_n(x)=\max\left(f_n(x_{n-1})+\langle \nabla f_n(x_{n-1}), x-x_{n-1}\rangle,0\right)
\] 
and let $W_N$ be the discrete-time work function. We obtain:
\begin{align*}
	\sum_{n=1}^N f_n(x_n)+||x_n-x_{n-1}|| &\leq \sum_{n=1}^N \tilde f_n(x_n)+||x_n-x_{n-1}||+\left(\sum_{n=1}^N f_n(x_n)-\tilde f_n(x_n)\right)\\
	&\leq (d+1)\min_x W_N(x)+\left(\sum_{n=1}^N f_n(x_n)-\tilde f_n(x_n)\right). 
\end{align*}

Suppose now that each $f_n$ is $L$-lipschitz. Then the equality $f_n(x_{n-1})=\tilde f_n(x_{n-1})$ implies $|f_n(x_n)-\tilde f_n(x_n)|\leq 2L||x_n-x_{n-1}||$. Because Theorem~\ref{thm:linearCR} and Proposition~\ref{prop:continuoustodiscrete} imply
\[
\sum_{n=1}^N ||x_n-x_{n-1}|| \leq d\min_x W_N(x),
\]
it follows that the resulting competitive ratio is at most $(2L+1)d+1$. Similar remarks apply to the result of Theorem~\ref{thm:dlogN}.
\end{rem}

\section{Competitive Ratio $O(\sqrt{d\log N})$ in Euclidean Space}
\label{sec:extra}

In this section we prove the discrete-time functional Steiner point has competitive ratio $O(\sqrt{d\log N})$ in Euclidean space (whose norm is denoted by $||\cdot||_2$). The same technique applies in any normed space given a suitable concentration result, however we restrict to the Euclidean case for convenience. The idea is as follows. Suppose that the average dual work function increase 
\[
	\dashint_{\theta\in \partial B_1^*}W^*_n(\theta)-W^*_{n-1}(\theta)\dd\mu(\theta)
\] 
at time-step $n$ is significant. Then by \eqref{eq:dual} the movement from $s(W_{n-1})\to s(W_{n})$ is an integral of pushes by different vectors $\theta$. By concentration of measure, these pushes decorrelate unless the total amount of pushing is exponentially small. 

\begin{lem}[{\cite[Lemma 2.2]{ball1997elementary}}]
\label{lem:concentration}
For any $0\leq \varepsilon <1$ and $|w|\leq 1$ in Euclidean space, the set \[\{{\theta\in \partial B_1}:\langle w,\theta\rangle \geq \varepsilon\}\] occupies at most $e^{-d \varepsilon^{2}/2}$ fraction of $\partial B_1$.
\end{lem}

\begin{lem}
\label{lem:hardtomove}

Suppose that $|W^*_n(\theta)-W^*_{n-1}(\theta)|\leq C$ for all $\theta\in \partial B_1$, and set 
\[
	\lambda=\dashint_{v\in B_1} W^*_n(v)-W^*_{n-1}(v)\dd v.
\]
Then the functional Steiner point movement is at most
\[
	||s(W_n)-s(W_{n-1})||_2 = O\left(\lambda\sqrt{d \left(1+\log\left(\frac{C}{\lambda}\right)\right)}\right).
\]
\end{lem}

\begin{proof}

Observe that 
\[
	||s(W_n)-s(W_{n-1})||_2=\max_{||w||_2=1}\langle w,s(W_n)-s(W_{n-1})\rangle.
\] 
Fixing a unit vector $w$, we estimate the inner product on the right-hand side. Set 
\begin{align*}
	g_n(\theta)&=W^*_n(\theta)-W^*_{n-1}(\theta)\geq 0,\\
	I_z&=\dashint_{\theta\in \partial B_1^*} g_n(\theta) \cdot 1_{\langle w,\theta\rangle\geq z} \dd\mu(\theta).
\end{align*}
Then $g_n(\theta)\in [0,C]$ for all $\theta$ and $\dashint_{\theta\in\partial B_1^*} g_n \dd\mu(\theta)=\lambda$. Consequently by Lemma~\ref{lem:concentration},
\begin{equation}
\label{eq:I-bound}
	I_z\leq \min\left(\lambda,C e^{-dz^2/2}\right).
\end{equation}
We thus find
\begin{align}
\nonumber
\langle w,s(W_n)-s(W_{n-1})\rangle &= d~\dashint_{\theta\in \partial B_1^*} g_n(\theta) \langle w,\theta\rangle \dd\mu(\theta)\\
\nonumber
&\leq d~\dashint_{\substack{\theta\in \partial B_1^*\\ \langle w,\theta\rangle\geq 0}} g_n(\theta) \langle w,\theta\rangle \dd\mu(\theta)\\
\nonumber
& = d\int_0^1 I_z \dd z\\
\label{eq:int-subgaussian}& \leq d\int_0^1 \min\left(\lambda,C e^{-dz^2/2}\right) \dd z.
% & \leq  O\left(\lambda C\sqrt{d \log\left(\frac{C}{\lambda}\right)}\right).
\end{align}
Here the second equality is the tail-sum integral formula. To estimate the resulting integral, set 
\[
	A=\sqrt{\frac{2\log (C/\lambda)}{d}}
\] 
so that $Ce^{-dA^2/2}=\lambda$. We will assume $A\leq 1$; if $A> 1$ then the expression \eqref{eq:int-subgaussian} is at most $d\lambda\leq dA\lambda$ and it suffices to mimic the below without the second term. We estimate 
\[
	\int_0^1 \min\left(\lambda,C e^{-dz^2/2}\right) \dd z= A\lambda+C\int_A^1 e^{-dz^2/2} \dd z.
\]
and use the simple bounds
\begin{align*}
	\int_A^1 e^{-dz^2/2} \dd z&\leq \int_0^1 e^{-dz^2/2} \dd z\leq O(d^{-1/2}),\\
	\int_A^1 e^{-dz^2/2} \dd z&\leq e^{-dA^2/2}\int_A^{\infty} e^{-dA(z-A)}\dd z = \frac{e^{-dA^2/2}}{dA}.
\end{align*}
Combining,
\begin{align*}
	\langle w,s(W_n)-s(W_{n-1})\rangle & \leq d\int_0^1 \min\left(\lambda,C e^{-dz^2/2}\right) \dd z\\
	&\leq dA\lambda+\min\left(C\sqrt{d},\frac{Ce^{-dA^2/2}}{A}\right)\\
	&=O\left(\lambda \sqrt{d \log\left(\frac{C}{\lambda}\right)}\right) + \min\left(C\sqrt{d},\lambda\sqrt{\frac{d}{2\log (C/\lambda)}}\right).
\end{align*}
With $u=\lambda/C\in [0,1]$, the last term is 
\begin{align*}
	C\sqrt{d}\cdot\min\left(1,\frac{u}{\sqrt{2\log(1/u)}}\right)
\end{align*}
For $u\leq [0,1/2]$, we have $\frac{u}{\sqrt{2\log(1/u)}}\leq O(u)$ giving the bound $O(\lambda \sqrt{d})$. For $u\geq 1/2$ we have $C\sqrt{d}\leq 2\lambda \sqrt{d}$. Hence in both cases,
\[
	\left\langle w,s(W_n)-s(W_{n-1})\right\rangle\leq O\left(\lambda\sqrt{d \left(1+\log\left(\frac{C}{\lambda}\right)\right)}\right)
\]
as desired. 

\end{proof}

\begin{thm}\label{thm:dlogN}

The discrete time functional Steiner point algorithm is $O(\sqrt{d\log N})$ competitive for chasing convex functions in Euclidean space.

\end{thm}

\begin{proof}

Call $(x_t)_{t\in [0,N]}$ the continuous path and $(x_{t_n})_{n\leq N}$ the discrete path for $t_n\in (n-1,n]$ as in Proposition~\ref{prop:continuoustodiscrete}. Since the service cost for the discrete path is at most that of the continuous path, we only need to establish the $O(\sqrt{d\log N})$ competitive ratio on the movement of the discrete path. By Lemma~\ref{lem:boundondual},
\[ 
	\max_{|\theta|\leq 1} W^*_N(\theta)\leq  2\cdot \min_{x}W_N(x).
\]
Set 
\[
	\lambda_n=\int_{\theta\in \partial B_1^*}W^*_{t_n}(\theta)-W^*_{t_{n-1}}(\theta)\dd\mu(\theta).
\]
Applying Lemma~\ref{lem:hardtomove} with $C=2\cdot \min_{x}W_N(x)$ to the movement $||x_{t_n}-x_{t_{n-1}}||_2$ at each step yields:
\begin{equation}\label{eq:discrete-movement}
	\sum_{n=1}^N ||x_{t_n}-x_{t_{n-1}}||_2\leq O(Cd^{1/2})\cdot \sum_{n\leq N} \frac{\lambda_n}{C} \sqrt{1+\log\left(\frac{C}{\lambda_n}\right)}.
\end{equation}
Here the values $\lambda_n$ are all non-negative and sum to $\dashint_{\theta\in\partial B_1^*} W^*_N(\theta)\dd\mu(\theta)\leq C$. Letting $h(u)=u\sqrt{1+\log(1/u)}$, one readily computes that for $u\in (0,1)$,
\begin{equation*}
	h'(u)=\frac{2\log(1/u)+1}{2(1+\log(1/u))^{1/2}}\geq 0,\quad\quad
	h''(u)=\frac{-2\log(1/u)-3}{4u(1+\log(1/u))^{3/2}}\leq 0.
\end{equation*} 
Jensen's inequality therefore implies that setting $\lambda_n=\frac{C}{N}$ for all $n\leq N$ in \eqref{eq:discrete-movement} gives an upper bound. It follows that the movement cost is at most $O(C\sqrt{d\log (N+1)})$.
\end{proof}

\section{Steiner Points of Level Sets}
\label{sec:unify}

\subsection{A Simplification for Chasing Convex Bodies}

Here we show that for chasing convex bodies in discrete time, it suffices to simply set $x_n=s(W_n)$ instead of reducing from a continuous-time problem via Proposition~\ref{prop:continuoustodiscrete}. This simplification does not seem possible for chasing convex functions. The movement cost estimates continue to hold with no changes in the proof, however establishing $s(W_n)\in K_n$ requires a short additional argument. Define the support set $\Supp(W)\subseteq \mathbb R_d$ of an abstract work function $W$ to be the set of points $x$ possessing a subgradient $v\in\nabla W(x)$ with $|v|<1$. 
% We refer to convex, $1$-Lipschitz functions $W:X\to\mathbb R^+$ with $W(x)-||x||$ uniformly bounded as \emph{work functions}. 
For a work function $W$ and convex body $K$, set
\[
	W^{K}(x)=\min_{y\in K} W(y)+||y-x||.
\]
If $W$ is the work function for some sequence of requests, then making an additional request of $K$ results in the new work function $W^K$.

\begin{prop}\label{prop:supp}

$\Supp(W^K)\subseteq K$ holds for any work function $W$ and convex body $K$.

\end{prop}

\begin{proof}
Suppose $x\notin K$ and set 
\[
	y\in\arg\min_{y_0\in K}(W(y_0)+||y_0-x||).
\]
For any $z$ on the segment $\overline{yx}$, it follows that $W(x)-W(z)=||x-z||$. This implies that no $v$ with $|v|<1$ can be a subgradient in $\nabla W_n(x)$.
\end{proof}

\begin{cor}
The algorithm $x_n=s(W_n)$ is $d$ competitive for chasing convex bodies, and $O(\sqrt{d\log N})$ competitive in Euclidean space.
\end{cor}

\begin{proof}
Proposition~\ref{prop:supp} and the primal definition \eqref{eq:primal} together imply $s(W_n)\in K_n$, i.e. the algorithm is valid. The $d$-competitiveness follows from Theorem~\ref{thm:linearCR} and the argument of Proposition~\ref{prop:continuoustodiscrete} while the $O(\sqrt{d\log N})$ competitive ratio in Euclidean space follows from the argument of Theorem~\ref{thm:dlogN}.
\end{proof}

\subsection{Steiner Points of Level Sets}

This final subsection has two main objectives. Theorem~\ref{thm:funcsteinerlevel} states that the functional Steiner point of any work function can be expressed as the Steiner point of large level sets. Corollary~\ref{cor:CBCsimple} states that the Steiner point of any level set of the work function $W_n$ is inside $K_n$ for convex body chasing. As we discuss at the end, Corollary~\ref{cor:CBCsimple} is related to the algorithm for chasing convex bodies given by \cite{AGGT}. Denote level sets by 
\[
	\Omega_{W,R}=\{x:W(x)\leq R\}.
\] 
It is easy to see that for any work function $W$ and $R\geq \min_x W(x)$,

\[
W^{\Omega_{W,R}}(x)=\Bigg\{\begin{array}{lr}
        W(x), & \text{for } x\in \Omega_{W,R}\\
        d(x,\Omega_{W,R})+R, & \text{for } x\notin \Omega_{W,R}.
        \end{array}
\]

\begin{thm}\label{thm:funcsteinerlevel}

For any work function $W$ and $R\geq \min_x W(x)$, it holds that $s(\Omega_{W,R})=s\left(W^{\Omega_{W,R}}\right)$ and $\lim_{R\to\infty} s(\Omega_{W,R})=s(W).$ Moreover if $\Supp(W)\subseteq \Omega_{W,R}$ then $s(\Omega_{W,R})=s(W)$.

\end{thm}

\begin{proof}

The dual definitions~\eqref{eq:dual0},~\eqref{eq:dual} imply
\begin{equation}\label{eq:levelset} 
	s(\Omega_{W,R})-s(W)= d~\dashint_{\theta\in \partial B^*_1} \bigg(W^*(\theta)+h_{\Omega_{W,R}}(\theta)\bigg) n(\theta) \dd \mu(\theta).
\end{equation}
Also for any $\theta\in\partial B_1^*$,
\begin{align*} 
	\left(W^{\Omega_{W,R}}\right)^*(\theta)&= \inf_{w\in X} \bigg(W^{\Omega_{W,R}}(w)-\langle w,\theta\rangle\bigg)\\
	&=\inf_{w\in \partial \Omega_{W,R}} \bigg(W^{\Omega_{W,R}}(w)-\langle w,\theta\rangle\bigg)\\
	&=R-h_{\Omega_{W,R}}(\theta).
\end{align*}
It follows from the symmetry $\theta\leftrightarrow -\theta$ that 
\[
\dashint_{\theta\in \partial B^*_1} n(\theta) \dd \mu(\theta)=0.
\]
Combining the above yields
\[
s(\Omega_{W,R})=s\left(W^{\Omega_{W,R}}\right).
\]
We proceed similarly for the second claim. For any $\theta\in\partial B_1^*$,
\begin{align*}
	W^*(\theta)&=\inf_{w\in X}(W(w)-\langle \theta,w\rangle)\\
	&=\lim_{R\to\infty}\inf_{w\in \Omega_{W,R}}(W(w)-\langle \theta,w\rangle)\\
	&=\lim_{R\to\infty}\inf_{w\in \partial\Omega_{W,R}}(W(w)-\langle \theta,w\rangle)\\
	&=\lim_{R\to\infty} \bigg(R-h_{\Omega_{W,R}}(\theta)\bigg).
\end{align*}
Because $W(x)-||x||$ is uniformly bounded it follows that the expression 
\[
	W^*(\theta)+h_{\Omega_{W,R}}(\theta)-R
\]
is uniformly bounded for $(\theta,R)\in (\partial B_1^*\times\mathbb R^+)$. As just shown it tends to $0$ as $R\to\infty$. The bounded convergence theorem therefore implies 
\[
	\lim_{R\to\infty}\dashint_{\theta\in\partial B_1^*} \left\vert W^*(\theta)+h_{\Omega_{W,R}}(\theta)-R \right\vert \dd\mu(\theta)=0.
\]

Combining with equation~\eqref{eq:levelset} shows that $\lim_{R\to\infty}||s(\Omega_{W,R})-s(W)||=0$, proving the second assertion. The last assertion is proved similarly after observing that $\Supp(W)\subseteq \Omega_{W,R}$ implies
\begin{align*}
	W^*(\theta)&=\inf_{w\in X}(W(w)-\langle \theta,w\rangle)\\
	&=\lim_{\lambda\uparrow 1}\inf_{w\in X}(W(w)-\langle \lambda\theta,w\rangle)\\
	&=\lim_{\lambda\uparrow 1}\inf_{w\in \Omega_{W,R}}(W(w)-\langle \lambda\theta,w\rangle)\\
	&=R-h_{\Omega_{W,R}}(\theta).
\end{align*}

\end{proof}

\begin{prop}\label{prop:unify} $\Supp(W^{\Omega_{W,R}})\subseteq \Supp(W)$ holds for any $R\geq \min_x W(x)$.

\end{prop}

\begin{proof}

Because $\Omega_{W,R}$ is a level set,
\[
	W^{\Omega_{W,R}}(x)=\Bigg\{\begin{array}{lr}
    W(x), & \text{for } x\in \Omega_{W,R}\\
    d(x,\Omega_{W,R})+R, & \text{for } x\notin \Omega_{W,R}\\
    \end{array}
 \]
Proposition~\ref{prop:supp} combined with the fact that $W$ and $W^{\Omega_{W,R}}$ agree inside $\Omega_{W,R}$ imply that the only possible new support points are on the boundary $\partial\Omega_{W,R}$. Fix a boundary point $y\in\partial\Omega_{W,R}\backslash\Supp(W)$. Because $y\notin \Supp(W)$, there exists a sequence $(y_i)_{i\in\mathbb N}\to y$ satisfying
\[
W(y)-W(y_i)\geq (1-o(1))||y-y_i||.
\]

Such a sequence of points $y_i$ must eventually satisfy $W(y_i)\leq W(y)$ and therefore $y_i\in\Omega_{W,R}$, implying $W(y_i)=W^{\Omega_{W,R}}(y_i)$. Hence 
\[
	W^{\Omega_{W,R}}(y)-W^{\Omega_{W,R}}(y_i)\geq (1-o(1))||y-y_i||.
\]
This implies $y\notin \Supp(W^{\Omega_{W,R}})$, completing the proof.

\end{proof}

\begin{cor}

\label{cor:unify}

Let $W=\widehat{W}^{K}$ for a work function $\widehat W$ and convex body $K$. For any $R\geq \min_x W(x)$, \[s(\Omega_{W,R})=s\left(W^{\Omega_{W,R}}\right)\in K.\]

\end{cor}

\begin{proof}

Propositions~~\ref{prop:supp} and \ref{prop:unify} show that 
\[
	\Supp\left(W^{\Omega_{W,R}}\right)\subseteq \Supp(W)\subseteq K.
\]
The primal definition~\eqref{eq:primal} of the functional Steiner point now implies $s(W^{\Omega_{W,R}})\in K$.

\end{proof}

\begin{cor}\label{cor:CBCsimple}

Let $W_n$ be the work function for convex body requests $(K_1,\dots,K_n)$. Then \[s(W_n^{\Omega_{W_n,R}})\in K_n\] for any $R\geq \min_x W_n(x)$. 

\end{cor}

\begin{proof}

Immediate from Corollary~\ref{cor:unify} with $\widehat W=W_{n-1}$ and $K=K_n$.

\end{proof}

\begin{rem}

\cite{AGGT} solved chasing convex bodies in Euclidean space by taking $x_n=s\left(W_n^{\Omega_{W_n,R_n}}\right)$ with $R_n=2^{\lceil \log_2(\min_x W_n(x))\rceil}$. This defines a selector by Corollary~\ref{cor:CBCsimple}. Estimating the movement cost is not difficult because the sets $W_n^{\Omega_{W_n,R}}$ decrease for fixed $R$. Note that $diam(\Omega_{W_n,R})\leq 2R$ because of the inequality $W_t(x)\geq ||x||$ (recall Proposition~\ref{prop:properties}). Using Theorem~\ref{thm:nested}, the movement from each fixed $R$ value is at most $O(\min(dR,R\sqrt{d\log T}))$. Summing over the geometric sequence of $R$ values yields the same upper bound as in Theorems~\ref{thm:linearCR} and \ref{thm:dlogN} up to a constant factor.

\cite{AGGT} prove that $s\left(W_n^{\Omega_{W_n,R_n}}\right)\in K_n$ using reflectional symmetries that may not exist in arbitrary normed spaces. Corollary~\ref{cor:CBCsimple} thus implies that their algorithm works for general norms.
\end{rem}

\section*{Acknowledgement}
The author thanks S\'ebastien Bubeck, Bo'az Klartag, Yin Tat Lee, and Yuanzhi Li for the introduction to convex body chasing and the Steiner point, and many stimulating discussions. He thanks Ethan Jaffe, Felipe Hernandez, and Christian Coester for discussions about properties of the work function, and the anonymous referee for several suggestions. He additionally thanks S\'ebastien for feedback on previous drafts and gratefully acknowledges the support of an NSF graduate fellowship and a Stanford graduate fellowship.

\bibliographystyle{alpha}
\bibliography{bib}

\appendix 

\section{Proof of Lemma~\ref{lem:derivativeWt}}

\begin{proof}

We prove the result for all $v\in B_1^*$ where $\nabla W^*_t(v)$ exists. This includes almost all $v$ by Alexandrov's theorem. Moreover it ensures the conjugate point $v_t^*=\arg\min_{w\in X} W(w)-\langle v,w\rangle$ is well-defined and that $W_t$ is strictly convex at $v_t^*$ \cite[Corollary 25.1.2]{rockafellar1970convex}. We write:
\begin{align*}
	W_{t+\delta}(v)=& \min_{x_s:[0,t+\delta]\to X} \left(\int_{0}^{t+\delta} (f_s(x_s)+||x'_s||\dd s -\langle v,x_{t+\delta}\rangle\right)\\
	=& \min_{ x_s:[t,t+\delta]\to X}\left(W_t(x_t)+\int_{t}^{t+\delta}f_s(x_s)+||x'_s||\dd s - \langle v,x_{t+\delta}\rangle\right)
\end{align*}

For small $\delta\in (0,\varepsilon)$, we show $W^*_{t+\delta}(v)=W^*_t(v)+\delta f_t(v^*_t)+o(\delta)$. For the upper bound, 
\begin{align*}
	W_{t+\delta}(v_t^*)&\leq  W_t(v_t^*)+\int_t^{t+\delta} f_s(v^*_t)\dd s\\ 
	&=  W_t(v_t^*)+\delta f_t(v^*_t)+o(\delta)
\end{align*}
holds by taking $x_s=v_t^*$ constant for $s\in [t,t+\delta)$ and recalling the assumption that $f_s(x)$ is continuous on $s\in [t,t+\delta)$. Since $v_t^*=\arg\min_x \big(W_t(x)-\langle x,v\rangle\big)$, the upper bound follows from
\begin{align*}
	W^*_{t+\delta}(v)&\leq  W_{t+\delta}(v_t^*)-\langle v,v_t^*\rangle\\
	&\leq W_t(v_t^*)+\delta f_t(v^*_t)+o(\delta)-\langle v,v_t^*\rangle\\
	&=W^*_t(v)+\delta f_t(v^*_t)+o(\delta).
\end{align*}

For the lower bound, the strict convexity of $W_t$ at $v_t^*$ implies
\[
	W_t(x)= W_t(v^*_t)+\langle v,x-v^*_t\rangle +\gamma(||x-v^*_t||)
\] 
where $\gamma:\mathbb R^+\to\mathbb R^+$ is continuous and increasing with unique minimum $F(0)=0$. Therefore any path $x_s:[0,t+\delta]\to X$ satisfies:
\[
W_t(x_t)+\int_{t}^{t+\delta}f_s(x_s)+||x'_s||\dd s - \langle v,x_{t+\delta}\rangle \geq W_t(v^*_t)+\langle v,x_t-v^*_t\rangle+\gamma(||x_t-v^*_t||) + \int_t^{t+\delta} f_s(x_s)+||x'_s||\dd s-\langle v,x_{t+\delta}\rangle.
\]
The observation $\int_t^{t+\delta} ||x'_s||\dd s \geq ||x_{t+\delta}-x_t|| \geq \langle v,x_{t+\delta}-x_t\rangle$ implies
\begin{align*}
	W_t(x_t)+\int_{t}^{t+\delta}f_s(x_s)+||x'_s||\dd s - \langle v,x_{t+\delta}\rangle\geq & W_t(x_t)-\langle v,v^*_t\rangle+f(||x_t-v^*_t||) + \int_t^{t+\delta} f_s(x_s)\dd s\\
	\geq & W_t(v^*_t)-\langle v,v^*_t\rangle +\gamma(||x_t-v^*_t||)+ \int_t^{t+\delta} f_s(x_s)\dd s\\
	\geq & W^*_t(v)+\gamma(||x_t-v^*_t||)+ \int_t^{t+\delta} f_s(x_s)\dd s.
\end{align*}
Because $W_{t+\delta}(v)=W_t(v)+O(\delta)$, we see that for $\delta\to 0$ small we must have $||x_t-v_t^*||=o_{\delta\to 0}(1)$ for any optimal trajectory $x_s$ witnessing the correct value $W_{t+\delta}$. Additionally,  
\[ 
	\int_t^{t+\delta} ||x'_s||\dd s + \langle v, x_t-x_{t+\delta}\rangle \geq (1-|v|)\int_t^{t+\delta} ||x'_s||\dd s \geq (1-|v|)\sup_{s\in [t,t+\delta]} |x_t-x_s|. 
\]
which similarly implies $\sup_{s\in [t,t+\delta]}||x_t-x_s||=o(1)$ for any optimal trajectory since $||v||<1$. It follows that all optimal trajectories satisfy $\int_t^{t+\delta} f_s(x_s)\dd s = \delta f_t(v^*_t) +o(\delta).$ This concludes the proof.

\end{proof}

\end{document}